\documentclass{amsart}
\usepackage{hyperref}
\usepackage{graphicx}
\usepackage{curves}
\usepackage{enumerate}
\usepackage{amsmath,amsthm,amssymb}
\usepackage{enumerate}
\usepackage{amsmath}
\usepackage{amssymb}
\usepackage{amsfonts}
\usepackage{amsthm}
\usepackage{epsfig}
\usepackage{graphicx}

\newtheorem{thm}{Theorem}[section]
\newtheorem{lem}{Lemma}[section]

\numberwithin{equation}{section}

\newcommand{\D}{\displaystyle}
\newcommand{\be}{\begin{equation}}
\newcommand{\ee}{\end{equation}}
\newcommand{\dd}{\textup{d}}
\newcommand{\ii}{\textup{i}}
\newcommand{\la}{\lambda}
\newcommand{\ord}{\mathrm{O}}
\renewcommand{\Im}{\operatorname{Im}}

\DeclareMathOperator{\res}{res}

\begin{document}
\title[Complete linearization of a mixed problem to the MB equations]{Complete linearization of a mixed problem to the Maxwell-Bloch equations by matrix Riemann-Hilbert problem}
 
\author[V. Kotlyarov]{Volodymyr Kotlyarov}
\address{Institute for Low Temperature Physics 47, Lenin ave, 61103 Kharkiv, Ukraine}
\email{\href{mailto:kotlyarov@ilt.kharkov.ua}{kotlyarov@ilt.kharkov.ua}}


\begin{abstract}
Considered in this paper the Maxwell-Bloch (MB) equations became known after Lamb
\cite{L1}-\cite{L4}. In  \cite{AKN} Ablowitz, Kaup and Newell proposed  the
inverse scattering transform (IST) to the Maxwell-Bloch equations for studying a physical phenomenon known as the self-induced transparency. A description of general solutions to the MB equations and their classification was done in \cite{GZM} by Gabitov, Zakharov and Mikhailov. In particular, they gave an approximate solution of the mixed problem to the MB equations in the domain $x,t\in(0,L)\times (0,\infty)$ and, on this bases, a description of  the phenomenon  of superfluorescence.   It was emphasized in \cite{GZM} that the IST method is non-adopted for the mixed problem.
Authors of the mentioned papers have developed the IST method in the form of the Marchenko integral equations. We propose another approach for solving the mixed problem to the MB equations in the quarter plane. We use a simultaneous spectral analysis of the both Lax operators and matrix Riemann-Hilbert (RH) problems. First, we introduce appropriate compatible solutions of the corresponding Ablowitz-Kaup-Newell-Segur (AKNS) equations and then we suggest such a matrix RH problem which corresponds to the mixed problem for MB equations. Second, we generalize this matrix RH problem, prove a unique solvability of the new RH problem and show that the RH problem (after a specialization of jump matrix) generates the MB equations. As a result we obtain solutions defined on the whole line and studied in \cite{AKN} and \cite{GZM}, solutions to the mixed problem studied below in this paper and solutions with a periodic (finite-gap) boundary conditions. The kind of solution is defined by the specialization of conjugation contour and jump matrix. Suggested matrix RH problems will be useful for studying the long time/long distance ($x\in\mathbb{R_+}$) asymptotic behavior of solutions to the MB equations using the Deift-Zhou method of steepest decent.
\end{abstract}
\maketitle

\section{Introduction}

The Maxwell-Bloch equations arise in different physical problems. Most significant applications
of this system deal with the problem of the propagation of an electromagnetic wave (ultrashort optical pulse) in a resonant medium with distributed two-level atoms. In particular, there are
the problem of self-induced transparency, the laser problems of  quantum amplifier and  super-fluorescence.
The system of the Maxwell-Bloch (MB) equations can be written in the form
\begin{alignat}{2}                                                   \label{MB1}
{\mathcal E}_t+{\mathcal E}_x =&\langle\rho\rangle,  \\
\rho_t+2\ii\lambda\rho=&{\mathcal N}{\mathcal E},\label{MB2}\\
{\mathcal N}_t =&-\frac{1}{2}({\mathcal E}^*\rho+{\mathcal E}\rho^*)\label{MB3}.
\end{alignat}
Here ${{\mathcal E}=\mathcal E}(t,x)$ is the complex electric field envelope, so that the field in the resonant medium is
\[
{\bf E}(t,x)={\mathcal E}(t,x)e^{\ii\Omega(x-t)}+{\mathcal E}^*(t,x)e^{-\ii\Omega(x-t)}.
\]
Subindexes mean partial derivatives in $t$ and $x$, and $*$ means a complex conjugation.
${\mathcal N}={\mathcal N}(t,x,\lambda)$ and $\rho=\rho(t,x,\lambda)$ are entries of the density matrix of a quantum  two-level atom subsystem. The parameter $\la$ is the deviation of the transition frequency of the given two-level atom from the mean frequency $\Omega$. The angular brackets mean  averaging
\be                                                                   \label{av}
\langle\rho\rangle=\Omega\int\limits_{-\infty}^\infty n(\lambda)\rho(t,x,\lambda)\dd\lambda \nonumber
\ee
with the given "weight" function $n(\lambda)$, such that
$$
\int\limits_{-\infty}^\infty  n(\lambda)\dd\lambda=\pm 1.
$$
The weight function $n(\lambda)$ characterizes  the inhomogeneous broadening, i.e. a form of the spectral line. It is the difference between the initial populations of the upper and lower levels. If  $n(\lambda)>0$, then an unstable medium is considered (the so-called quantum laser amplifier). If  $n(\lambda)<0$, then a stable medium is considered (the so-called attenuator).

The mixed problem to the MB equations is defined by initial and boundary conditions
\begin{alignat}{2}                                        \label{ic}
{\mathcal E}(0,x)=&{\mathcal E}_0(x),  & \qquad 0<x<L\le\infty\\
\rho(0,x,\lambda)=&\rho_0(x,\lambda),\\
{\mathcal N}(0,x,\lambda)=&{\mathcal N}_0(x,\lambda),\label{ic3}\\
                                       \label{bc}
{\mathcal E}(t,0)=&{\mathcal E}_{in}(t),  & \qquad 0<t<\infty.
\end{alignat}

We assume that entering in the medium $0<x<L$ pulse ${\mathcal E}_{in}(t)$
is smooth and fast decreasing. Initial functions ${\mathcal E}_0(x)$ and $\rho_0(x,\lambda)$ are also smooth for $0<x<L$ and $\la\in\mathbb{R}$.
Function ${\mathcal N}_0(x,\la)$ is defined by $\rho_0(x,\lambda)$:
\be    \label{no}
{\mathcal N}_0(x,\lambda)= \sqrt{1-|\rho_0(x,\lambda)|^2},\qquad  0< x<L \le\infty.
\ee
Here we chose the positive branch of the square root.
To find (\ref{no}) we use the second and the third equations of (\ref{MB1})-(\ref{MB3}). They give
\[
\frac{\partial}{\partial t}({\mathcal N}^2(t,x,\lambda)+|\rho(t,x,\lambda)|^2)=0,
\]
and one can put
\be\nonumber
{\mathcal N}^2(t,x,\lambda)+|\rho(t,x,\lambda)|^2\equiv 1.
\ee
In what follows we also put $\Omega=1$ in (\ref{av}).

The Lax pair for the Maxwell-Bloch system was first found in \cite{AKN} by using
results of \cite{L1}-\cite{L4} (see also \cite{GZM} and \cite{AS}).
It was shown that (\ref{MB1})-(\ref{MB3}) are equivalent to the over determined linear system, known as Ablowitz-Kaup-Newell-Segur (AKNS) equations:
\begin{align}                                                    \label{teq}
w_t+\ii\lambda\sigma_3 w&=-H(t,x) w,\\                                      \label{xeq}
w_x-\ii\lambda\sigma_3 w +\ii G(t,x,\lambda) w&=H(t,x) w,
\end{align}
where $\sigma_3,  H(t,x), G(t,x,\lambda)$ are the matrices
\begin{align*}
&\sigma_3 =
\begin{pmatrix}1&0\\0&-1\end{pmatrix},\quad H(t,x)=\frac{1}{2}
\begin{pmatrix}0&{\mathcal E}(t,x)\\
-\bar{\mathcal E}(t,x)&0\end{pmatrix}, \notag\\
&G(t,x,\lambda)= {\rm p.v.}\frac{1}{4}\int\limits_{-\infty}^\infty
\begin{pmatrix}
  {\mathcal N}(t,x,s) &\rho(t,x,s) \\
  \rho^*(t,x,s) &-{\mathcal N}(t,x,s) \\
\end{pmatrix}\frac{n(s)}{s-\lambda}\dd s,\label{G}
\end{align*}
and the symbol $\rm{p.v.}$ denotes the Cauchy principal value
integral. Differential equations (\ref{teq}) and (\ref{xeq})
are compatible if and only if  ${\mathcal E}(t,x)$,
$\rho(t,x,\lambda)$ and ${\mathcal N}(t,x,\lambda)$ satisfy equations
(\ref{MB1})-(\ref{MB3}) (cf. \cite{AS}). As shown in \cite{L4},
$\rho(t,x,\lambda)$ and ${\mathcal N}(t,x,\lambda)$ are related to the fundamental solutions of (\ref{teq}). Indeed, let $\varPhi(t,x,\lambda)$ be a
solution of (\ref{teq}) such that $\det\varPhi(t,x,\lambda)\equiv1$. Then it is easy to show that $F(t,x,\lambda)=\varPhi(t,x,\lambda)\sigma_3\varPhi^\dag(t,x,\lambda)$,
where $\varPhi^\dag$ is the Hermitian-conjugated to $\varPhi$,
satisfies equation \be\label{F}
F_t+[\ii\la\sigma_3+H,F]=0,\qquad F(t,x,\lambda)=\begin{pmatrix}
  {\mathcal N}(t,x,\lambda) &\rho(t,x,\lambda) \\
  \rho^*(t,x,\lambda) &-{\mathcal N}(t,x,\lambda) \\
\end{pmatrix}
\ee
which is a matrix representation of equations (\ref{MB2})-(\ref{MB3}).
It is well-known that the fundamental solution of equation (\ref{teq}) is entire in $\la$. Therefore the matrix $F(t,x,\la)=\varPhi(t,x,\la)\sigma_3\varPhi^\dag(t,x,\la)$
is smooth in $\la\in\mathbb{R}$ if initial matrix $F(0,x,\la)$ is smooth too.

In what follows we will use the $x^\pm$-equations (upper/lower
bank equations instead of $x$-equation):
\be\label{pmxeq}
w_x-\ii\lambda\sigma_3 w +\ii G_\pm(t,x,\lambda) w=H(t,x) w, \ee
where \be\nonumber
G_\pm(t,x,\lambda)=\frac{1}{4}\int\limits_{-\infty}^\infty
\frac{F(t,x,s) n(s)}{s-\lambda\mp\ii0}d s={\rm
p.v.}\frac{1}{4}\int\limits_{-\infty}^\infty
\D\frac{F(t,x,s)n(s)}{s-\la
}ds\pm\frac{\pi\ii}{4}F(t,x,\la)n(\la). \ee

Thus we have new Lax pairs ($t$- and $x^+$-equations and $t$- and $x^-$-equations)
for the MB equations.  Equations (\ref{teq}) and
(\ref{pmxeq}) (as well as (\ref{teq}) and (\ref{xeq})) are
compatible if and only if   ${\mathcal E}(t,x)$,
$\rho(t,x,\lambda)$ and ${\mathcal N}(t,x,\lambda)$ satisfy equations
(\ref{MB1})-(\ref{MB3}). This known fact is proved below (see section 2).

Main goal of the paper is to show that the inverse
scattering transform is applicable to the  mixed problem for the MB
equations using a simultaneous spectral analysis of the compatible
differential equations (\ref{teq}) and (\ref{pmxeq}). We develop the IST method in the form of matrix Riemann-Hilbert problem in the complex $z-$plane and give an integral representation for ${\mathcal E}(t,x )$ through the solution of a singular integral equation, which is equivalent to the matrix RH problem. To produce this RH problem we introduce appropriate compatible solutions of the corresponding AKNS equations and input spectral (scattering) data defined through given initial and boundary conditions for MB equations. Then we suggest such a matrix RH problem which corresponds to the mixed problem for MB equations. Further, we give a formulation of more general matrix RH problem, prove a unique solvability of the new RH problem. After a specialization of jump matrix we prove that the RH problem  generates the AKNS linear equations for the MB equations with a given inhomogeneous broadening. Thus this RH problem generates different solutions to the MB equations. There are solutions defined on the whole line and studied in \cite{AKN} and \cite{GZM}, solutions to the mixed problem (\ref{ic})-(\ref{bc}) with vanishing boundary conditions and periodic (finite-gap) boundary conditions, etc. The kind of solution is defined by the specialization of conjugation contour and jump matrix. Suggested matrix RH problems will be useful for studying the long time/long distance ($x\in\mathbb{R_+}$) asymptotic behavior of solutions to the MB equations using the nonlinear method of steepest decent.

Our approach differs from that was proposed in \cite{Kis} for the Goursat problem to the MB equations.
We develop the method initially proposed in \cite{F1}-\cite{FI2} and also in \cite{BMK}-\cite{MK}.
This method uses the simultaneous spectral analysis of the Lax operators.
We formulate such a matrix RH problem which allow to prove that the mixed problem is linearizable completely.

We will  use the following notations: if $A$ denotes  matrix
$A=(A[1], A[2])$, then  vectors $A[1]$, $A[2]$ denote
the first and the second columns of matrix $A$. We also set
$[A,B]=AB-BA$.
\section{Compatible solutions of the AKNS equations}
\setcounter{equation}{0}

Let
\begin{align}                                                  \label{U}
W_t&=U(t,x,\lambda) W,\\                                              \label{V}
W_x&=V(t,x,\lambda) W
\end{align}
be a system of matrix differential equations.
Equations (\ref{U}),  (\ref{V}) are compatible if $W_{tx}\equiv
W_{xt}$ for any  solution of these equations, i.e.\
the following condition holds~:
\begin{equation}                                              \label{ZCC}
U_x-V_t+UV-VU=0.
\end{equation}

\begin{lem} \label{lem 2.1}
Let  $U(t,x,\lambda)$ and $V(t,x,\lambda)$   are defined and smooth for all
$t,x,\lambda \in\mathbb{R}$ and satisfy condition $(\ref{ZCC})$.
Let $W(t,x,\lambda)$ satisfy $t$-equation $(\ref{U})$ for all $x$, and
let $W(t_0,x,\lambda)$ satisfy  $x$-equation $(\ref{V})$ for some
$t=t_0$, $t_0\le\infty$. Then  $W(t,x,\lambda)$ satisfies $x$-equation for all $t$. The same result is true, if one changes $t$ into $x$ and vice versa.
\end{lem}

\begin{proof}
Let $\widehat{W}(t,x,\lambda)= W_x-V(t,x,\lambda) W$. Then $\widehat{W}(t,x,\lambda)$ solves  equation (\ref{U}).  Indeed
\[
\widehat{W}_t=U(t,x,\lambda)\widehat{W}+( U_x-V_t+[U,V]) W =U(t,x,\lambda)\widehat{W}.
\]
Matrices $W$ , $\widehat{W}$ are solutions of the same equation. Therefore they are linear dependent, i.e.
$
\widehat{W}(t,x,\la)= W(x,t,\la)C(x,\la).
$
Since $\widehat{W}(t_0,x,\la)=0$, then $C(x,\la)$ and $\widehat{W}(t,x,\la)$ are
identically equal to zero. Lemma \ref{lem 2.1} is proved.
\end{proof}

Due to (\ref{teq}), (\ref{xeq}) and (\ref{pmxeq})  we have three AKNS systems. The $t$-equation (\ref{teq}) is defined by matrix $U(t,x,\la)=-\ii\la\sigma_3-H(t,x)$ while the $x$-equation (\ref{xeq}) and upper/lower bank $x^\pm$-equations (\ref{pmxeq}) are defined by matrices:
\begin{align*}
V(t,x,\la)=&\ii\la\sigma_3+H(t,x)-\ii G(t,x,\la),\\
V^\pm(t,x,\la)=&V(t,x,\la)\pm\D\frac{\pi\ii n(\la)}{4}F(t,x,\la).
\end{align*}
The compatibility condition of  $t$-equation (\ref{teq}) and $x$-equation
(\ref{xeq}) takes the form of (\ref{ZCC}). The compatibility condition of the $t$-equation (\ref{teq}) and $x^\pm$-equations (\ref{pmxeq}) is as follows:
\[
U_x-V^\pm_t+[U,V^\pm]=U_x-V_t+[U,V]\mp\frac{\pi\ii n(\la)}{4}(F_t+[\ii\la\sigma_3+H,F])=0.
\]
If $F_t+[\ii\la\sigma_3+H,F]=0$ (in view of eq.(\ref{F})), then the new compatibility conditions coincides with (\ref{ZCC}). Conversely, if $U_x-V_t+[U,V]=0$, then  $F_t+[\ii\la\sigma_3+H,F]=0$ as well as $U_x-V^\pm_t+[U,V^\pm]=0$. Each of them gives the same MB equations. The upper/lower bank equation (\ref{pmxeq}) (in a contrast with $x$-equation (\ref{xeq})) will allow to obtain such vector-solutions of the AKNS equations, which have analytic continuation to the upper and lower complex $z$-plane ($z=\la+\ii\nu$). We will use them for a construction of the matrix Riemann-Hilbert problem, which gives a solution of the mixed problem to the Maxwell-Bloch equations.

\begin{thm}
  Let  ${\mathcal E}(t,x )$, ${\mathcal N}(t,x,\la)$,  $\rho(t,x,\la)$ be a smooth   solution to the MB equations (\ref{MB1})-(\ref{MB3}) defined for all $t,x\in\mathbb{R}_+$ and  $\la\in\mathbb{R}$. Let ${\mathcal E}(t,0)$ is fast decreasing as $t\to+\infty$ and let $\rho(0,x,\la)\equiv0$ for $x>L$ or is also fast decreasing as $x\to+\infty$ ( if $L=\infty$). Then there exists two pairs of compatible solutions $Y^\pm(t,x,\la)$ and $Z^\pm(t,x,\la)$ of the AKNS equations (\ref{teq}), (\ref{pmxeq}) such that
\begin{equation}                                         \label{Y}
Y^\pm(t,x,\la)= W^\pm(t,x,\la)\Phi(t,\la),
\end{equation}
\begin{equation}                                         \label{Z}
Z^\pm(t,x,\la)= \Psi(t,x,\la)w^\pm(x,\la).
\end{equation}
Here $W^\pm(t,x,\la)$ satisfies $x^\pm$-equation (\ref{pmxeq}) for all $t$, $W(t,0,\la)=I$, and $\Phi(t,\la)$  satisfies  $t$-equation (\ref{teq}) with $x=0$ under the Jost initial condition:
$$
 \lim\limits_{t\to\infty}\Phi(t,\la)e^{\ii\la t\sigma_3}=I.
$$
Function $\Psi(t,x,\la)$ satisfies  $t$-equation (\ref{teq}) under initial condition $\Psi(0,x,\la)=I$, and $w^\pm(x,\la)$ satisfies $x^\pm$-equation (\ref{pmxeq}) with $t=0$ under initial condition $w^\pm(L,\la)=e^{\ii
L\eta_\pm(\la)\sigma_3}$, where
\be\label{etapm}
\eta_\pm(\la)=\la-\frac{1}{4}\int\limits_{-\infty}^\infty
\D\frac{n(s) ds}{s-\la\mp\ii 0}.
\ee
If $L=\infty$ then the condition at $x=L$ changes with the Jost initial condition:
$$
\lim\limits_{x\to\infty} w^\pm(x,\la)e^{-\ii\eta_\pm(\la)x\sigma_3}=I.
$$

\begin{proof}
If matrices $W^\pm(t,x,\la)$, $\Phi(t,\la)$ and
$\Psi(t,x,\la)$, $w^\pm(x,\la)$ do exist then, due to Lemma 2.1, the products $W^\pm(t,x,\la)\Phi(t,\la)$ and $\Psi(t,x,\la)w^\pm(x,\la)$ are  compatible solutions  of the AKNS equations (\ref{teq}), (\ref{pmxeq}). The existence of these matrices will be given in the next Lemmas.
\end{proof}
\end{thm}
\begin{figure}[ht]
\begin{center}
\epsfig{width=60mm,figure=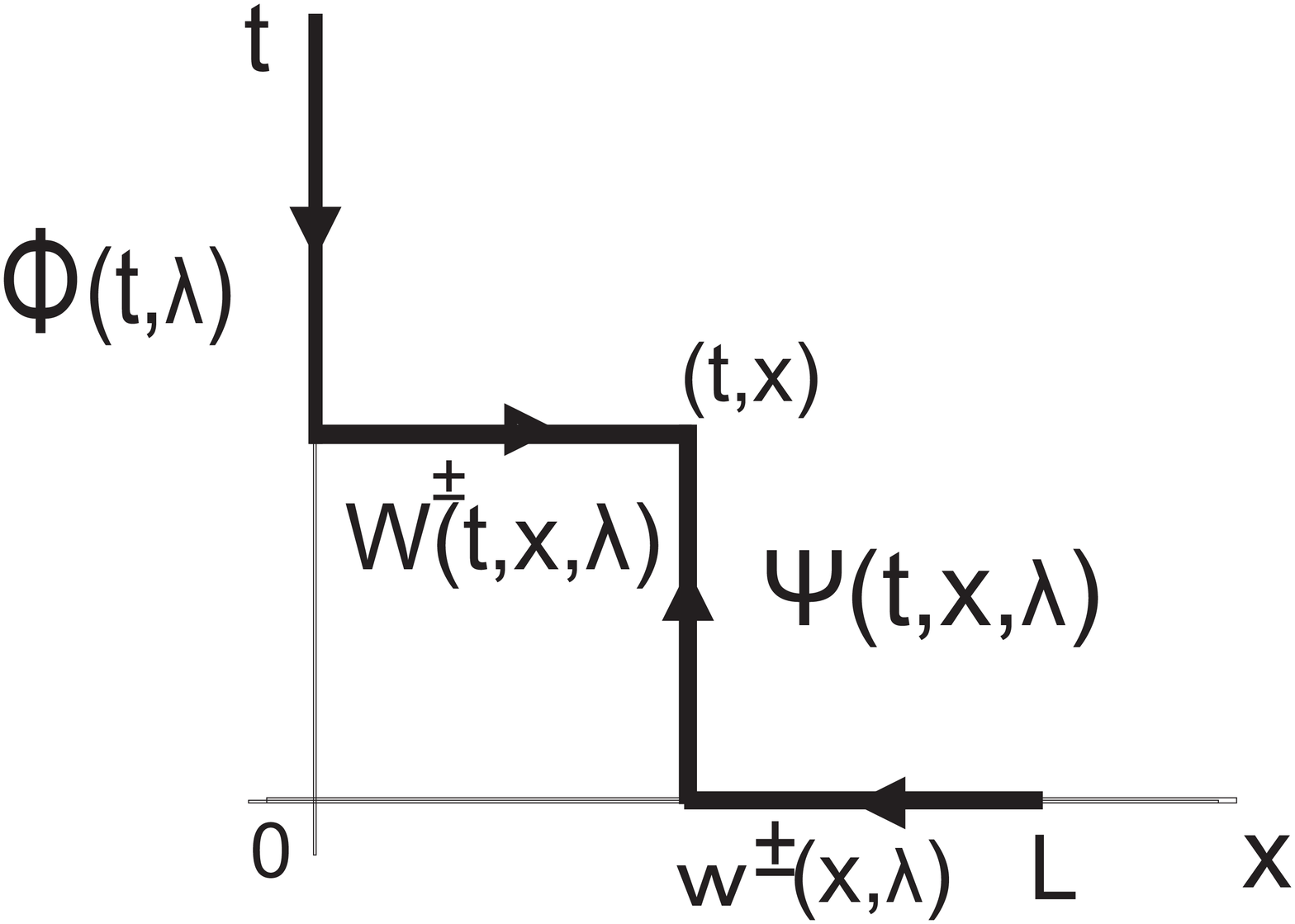 }
\end{center}
\end{figure}
\begin{lem}                                             \label{lem 2.2}
Let ${\mathcal E}(t,0)={\mathcal E}_1(t)$ be smooth and fast decreasing as $t\to \infty$. Then the Jost solution  $\Phi(t,\la)$ has  an integral representation:
\begin{equation}                                             \label{varPhi}
\Phi(t,\la)=e^{-\ii\la t\sigma_3}+\int_t^\infty K(t,\tau)e^{-\ii\la\tau\sigma_3}\dd\tau, \qquad
\Im\la=0.
\end{equation}
The kernel $K(t,\tau)$ satisfies the symmetry condition
$ {K}^*(t,\tau)= \Lambda  K(t,\tau) \Lambda^{-1}$
with matrix $ \Lambda=\begin{pmatrix}0&1\\-1&0
\end{pmatrix}$  and
$
[\sigma_3, K(t,t)]=\sigma_3 H(t,0).
$
\noindent The kernel $K (t,\tau)$ is smooth and fast decreasing as $t+\tau\to \infty$.
\end{lem}

The proof of this Lemma is well-known (cf.\ \cite{FT}). Due to this integral representation, vector-columns $\Phi[1](t,\la)$ and $\Phi[2](t,\la)$ of the matrix $\Phi(t,\la)=(\Phi[1](t,\la), \Phi[2](t,\la))$ have  analytic continuations $\Phi[1](t,z)$ and $\Phi[2](t,z)$ to the lower and upper half planes of the
complex $z$-plane respectively.  The next asymptotic formulae are valid:
\begin{itemize}
\item  $ \Phi[1](t,z)e^{\ii z t}=\begin{pmatrix}
                                  1 \\
                                  0
                                \end{pmatrix} +O(z^{-1}), \quad \Im z\le0,\quad z\to\infty,
                                $
\item $\Phi[2](t,z)e^{-\ii z t}=\begin{pmatrix}
                                  0 \\
                                  1
                                \end{pmatrix} +O(z^{-1}), \quad\Im z\ge0\quad z\to\infty$.
\end{itemize}
They are a simple consequence of integral representation (\ref{varPhi}).
\begin{lem}                                            \label{lem 2.3}
Let $n(\la)$ be a H\"older function for $\la\in\mathbb{R}$ and
let ${\mathcal E}(t,x)$, ${\mathcal N}(t,x,\la)$,  $\rho(t,x,\la)$ be smooth functions
for $t,\la\in\mathbb{R}$ and $x\in\mathbb{R}_+$.
Then the solution $W^\pm(t,x,\la)$ can be represented in the form:
\be\label{W}
W^\pm(t,x,\la)= e^{\ii\la x\sigma_3}\chi^\pm(t,x,\la),\ee
where $\chi^\pm(t,x,\la)$ is the unique solution of the Volterra integral
equation
\begin{equation}      \label{hatW}
\chi^\pm(t,x,\la)=I +\int\limits_0^x
e^{-\ii\la y \sigma_3}(H(t,y)-\ii G_\pm(t,y,\la))e^{\ii\la y\sigma_3}\chi^\pm(t,y,\la)
dy,
\end{equation}
The solutions $W^\pm(t,x,\la)$ are smooth in $t$ and $x$, related each other by the
formula:
$$
W^{-}(t,x,\la)=\sigma_2W^{+*}(t,x,\la)\sigma_2,\qquad \sigma_2=\begin{pmatrix}
0&\ii\\-\ii&0\end{pmatrix}
$$ and have an analytic continuations $W(t,x,z)$ to the upper $\mathbb{C}_+$ and  $\overline W(t,x,z)$ to the lower $\mathbb{C}_-$ complex $z$-plane respectively. Moreover,  $\overline W(t,x,z)e^{-\ii zx}$ is continuous, bounded in $\mathbb{C}_-\cup\mathbb{R}$ and have the following asymptotics:
\[
\overline  W (t,x,z)e^{-\ii zx}=\begin{pmatrix}
                                  1&0 \\
                                  0& e^{-2\ii zx}
\end{pmatrix}+O(z^{-1}), \qquad \Im z\le0,
\quad z\to\infty,
\]
and matrix $W(t,x,z)e^{\ii zx}$ is continuous, bounded
in $\mathbb{C}_+\cup\mathbb{R}$ and have the following asymptotics:
\[
 W (t,x,z)e^{\ii zx}=\begin{pmatrix}
                               e^{2\ii zx} & 0\\
                               0&1
\end{pmatrix}+O(z^{-1}), \qquad \Im z\ge0.
\quad z\to\infty.
\]
\end{lem}
\begin{proof}The solvability of the Volterra integral equation (\ref{hatW}) and
the smoothness of the solution with respect to $t$ and $x$ can be easily proved by using  the method of successive approximations. Solutions $W^\pm(t,x,\la)$  are not independent. One can easily verify that $W^{-}(t,x,\la)=\sigma_2W^{+*}(t,x,\la)\sigma_2$.
Since $ G_\pm(t,x,\la)$   has an analytic continuation to the
domain $\mathbb{C}_\pm$, the matrix $ W^\pm(t,x,\la)$ has also an analytic continuation for $z=\la+\ii\nu\in \mathbb{C}_\pm$ which we denote as $ W(t,x,z)$ for $z\in\mathbb{C}_+$ and $ \overline W(t,x,z)$ for $z\in\mathbb{C}_-$.  Thus the matrix $\overline W(t,x,z)e^{-\ii zx}$ ($W(t,x,z)e^{\ii zx}$) is analytic,
continuous and bounded in $\mathbb{C}_-\cup\mathbb{R}$ ($\mathbb{C}_+\cup\mathbb{R}$) and has the above asymptotic behavior.
\end{proof}

The previous lemmas and
formulas (\ref{Y}), (\ref{varPhi}), (\ref{W}), (\ref{hatW})
imply the following properties of the matrix
$Y^\pm(t,x,\la)=(Y^\pm[1](t,x,\la)\quad Y^\pm[2](t,x,\la))$:
\begin{enumerate}
\item  $Y^\pm(t,x,\la)$  satisfy the $t-$ and $x^\pm-$equations;
\item $\det Y^\pm(t,x,\la) \equiv1, \quad \la\in\mathbb{R} $;
\item the map $(x,t)\longmapsto Y^\pm(t,x,\la)$  is smooth in $t$
and $x$;
\item vector column $Y^-[1](t,x,\la)$ has analytic continuation
$\overline Y[1](t,x,z)$ for $z\in\mathbb{C}_-$, which is continuous and bounded in $z\in\mathbb{C}_-\cup\mathbb{R}$ and
\begin{align*}
\overline Y[1](t,x,z)e^{\ii z t-\ii zx }=&
\begin{pmatrix}1\\0\end{pmatrix}+ O(z^{-1}), \quad
&z\to\infty;\end{align*}
\item vector column $Y^+[2](t,x,\la)$ has analytic continuation   $Y[2](t,x,z)$ for $z\in\mathbb{C}_+$, which is continuous and bounded in $z\in\mathbb{C}_+\cup\mathbb{R}$ and
\begin{align*}
Y[2](t,x,z)e^{-\ii z t+\ii zx }=&
\begin{pmatrix}0\\1\end{pmatrix}+ O(z^{-1}), &z\to\infty.
\end{align*}
\end{enumerate}


\begin{lem}  \label{lem 2.4}
Let ${\mathcal E}(t,x)$  be smooth functions in the domain of its definition.
Then the function $\Psi(t,x,\la)$ has an integral representation:
\begin{equation}                                             \label{phi}
\Psi(t,x,\la)=e^{-\ii\la t\sigma_3}+\int_{-t}^t
L(t,\tau,x)e^{-\ii\la\tau\sigma_3}d\tau.
\end{equation}
The kernel $L(t,\tau,x)$ is smooth,  satisfies the symmetry
condition $ L^*(t,\tau,x)= \Lambda L(t,\tau,x) \Lambda^{-1}$ with
matrix $ \Lambda=\begin{pmatrix}0&1\\-1&0
\end{pmatrix}$  and
$
[\sigma_3, L(t,t,x)]=H(t,x)\sigma_3. \nonumber
$
\end{lem}
The proof is well-known (sf.\cite{BMK}). This lemma gives that $\Psi(t,x,\la)$ has an analytic continuation $\Psi(t,x,z)$ for all $z\in\mathbb{C}$. Moreover
$\Psi(t,x,z)e^{\ii zt}$ is continuous, bounded in $\mathbb{C}_+\cup\mathbb{R}$ and
\be\label{psias1}
\Psi(t,x,z)e^{\ii zt}=\begin{pmatrix}
1&0\\0&e^{2\ii zt}
\end{pmatrix}+O(z^{-1}),\qquad |z|\to\infty,
\ee
and $\Psi(t,x,z)e^{-\ii zt}$ is continuous,  bounded in $\mathbb{C}_-\cup\mathbb{R}$ and
\be\label{psias2}
\Psi(t,x,z)e^{-\ii zt}= \begin{pmatrix}
e^{-2\ii zt}&0\\0&1
\end{pmatrix} +O(z^{-1}),\qquad |z|\to\infty.
\ee

If $H(t,x)\equiv0$ and $F(t,x, \la)\equiv \sigma_3$, then
$x^\pm$-equation has the exact solution
$e^{\ii x\eta_\pm(\la)\sigma_3}$,  where $\eta_\pm(\la)$ is defined in (\ref{etapm}). These functions are the boundary values of the sectionally analytic function
\be\label{eta}
\eta(z)=z-\frac{1}{4}\int\limits_{-\infty}^\infty \frac{n(s)}{s-z}ds
\ee
for $z\in\mathbb{C}_\pm$.  Across the real line it has the jump:
\[
 \eta(\la+\ii0)-\eta(\la-\ii0)=-\frac{\pi\ii}{2} n(\la)\qquad \la\in\mathbb{R}.
\]
Furthermore, for $z=\la+\ii\nu$, 
\[
\Im\eta(z)=\nu\left(1-\D\frac{1}{4}\int\limits_{-\infty}^\infty\D \frac{n(s)\dd s}
{(s-\la)^2+\nu^2}\right)=\nu(1- I(\la,\nu)).
\]
If $n(\la)<0$ then ${\rm sign}\Im\eta(z)={\rm sign}\Im z$.
In the case $n(\la)>0$,
\begin{align*}
\D \frac{\Im\eta(z)}{\Im z}&> 0   \qquad \rm{if} \qquad   I(\la,\nu)<1;\\
\D \frac{\Im\eta(z)}{\Im z}&< 0   \qquad \rm{if} \qquad   I(\la,\nu)>1.
\end{align*}
For $n(\la)=\delta(\la)$ this integral equals to $1/4(\la^2+\nu^2)=1/|2z|^2$ and equation $  I(\la,\nu)=1$ defines the circle $|z|=1/2$.
In a general position when $n(\la)>0$, there exists a  curve, which is
the boundary of a domain $D$ containing the origin of the complex $z$-plane.
This  curve is defined by equation: $\Im\eta(z)=0$ when $\Im z\neq0$. It is symmetric with respect to real $\la$-axis, because $\eta^*(z^*)=\eta(z)$. We denote it as $\gamma\cup\overline\gamma$, where $\gamma$ lies in $\mathbb{C}_+$, and $\overline\gamma$ in $\mathbb{C}_-$. Thus there is such a domain  $D=D^+\cup D^-$  that
$${\rm sign}\Im\eta(z)=\begin{cases}-{\rm sign}\Im z, \qquad z\in D=D^+\cup D^-,\\
{\rm sign}\Im z, \qquad z\in(\mathbb{C}_+\setminus D^+)\cup(\mathbb{C}_-\setminus D^-).
\end{cases}$$
Domain $D$ may be bounded and unbounded as well. Indeed, let
$$
n(\la)=\begin{cases}\D\frac{1}{2\varepsilon},\qquad |\la|\le\varepsilon,\\\\
0, \qquad\quad |\la|>\varepsilon.  \end{cases}
$$
Then the curve $\gamma$ ($\Im\eta(z)=0, \Im z\neq0$) is described by equation:
$$
\la=\la(\nu)=\sqrt{\varepsilon^2-\nu^2+\D\frac{2\varepsilon\nu}{\tan8\varepsilon\nu}},
$$
where the square root is positive. If $\varepsilon$ is positive and sufficiently small then there exists $\delta>0$ such that
$$
\sqrt{\varepsilon^2-\nu^2+\D\frac{1}{4+\delta}}<\la(\nu)<
\sqrt{\varepsilon^2-\nu^2+\D\frac{1}{4}},
$$
i.e. finite curve $\gamma$ together with interval $[\la_-,\la_+]$ ($\la_\pm=\pm\sqrt{1/4+\varepsilon^2}$) bound domain $D^+$.
An example of the unbounded domain $D$ gives a physical and simple model with the Lorentzian line shape:
$$
n(\la)=\D\frac{l}{\pi}\frac{1}{\la^2+l^2},  \qquad l>0.
$$
Then
$$
\eta(z)=z-\D\frac{1}{4}\int\limits_{-\infty}^\infty\D\frac{n(\la)}{\la-z}d \la=
\begin{cases}z+\D\frac{1}{4(z+\ii l)}, \qquad z\in\mathbb{C}_+\\
z+\D\frac{1}{4(z-\ii l)}, \qquad z\in\mathbb{C}_-
\end{cases}
$$
and curve $\gamma$ is defined by equation:
$$
\la^2=\D\frac{\nu+l}{4\nu}-(\nu+l)^2.
$$
This equation yields that $\la_\pm(0)=\pm\infty$, i.e. $(\la_-,\la_+)=\mathbb{R}$, and $\nu_{\max}=(\sqrt{1+l^2}-l)/2$. Thus domain $D^+$ is unbounded along the real line because $\la(\nu)$ is unbounded when $\nu\to0$. Domain $D^-$ being complex conjugated to $D^+$ is unbounded also. 

If $F(0,x, \la)= \sigma_3$ for $x> L$ we put
\[
\hat H^\pm(x,\la)=H(0,x)-\ii\left(G_\pm(0,x,
\la)-g_\pm(\la)\sigma_3\right),\qquad 0\le x\le L,
\]
\[
g_\pm(\la)= \frac{1}{4}\int\limits_{-\infty}^\infty\D \frac{n(s)
ds}{s-\la\mp\ii 0}.
\]

\begin{lem}                                            \label{lem 2.5}
Let  $n(\la)$ be a H\"older function and fast decreasing as $\la\to\pm\infty$.
Let ${\mathcal E}(0,x)={\mathcal E}_0(x)$, ${\mathcal N}(0,x,\la)={\mathcal N}_0(x,\la)$,
$\rho(0,x,\la)=\rho_0(x,\la)$ be smooth for $\la\in\mathbb{R}$, $0\le x\le L$ and
fast decreasing as $x\to\infty$ when $L=\infty$.
Then  $w^\pm(x,\la)$ can be represented in the form:
\be\label{w} w^\pm(x,\la)= \hat w^\pm(x,\la)e^{\ii\eta_\pm(\la)x\sigma_3},\ee
where $\hat w^\pm(x,\la)$ is the unique solution of the Volterra integral
equation
\begin{equation}      \label{hatw}
\hat w^\pm(x,\la)=I -\int\limits_x^L
e^{\ii\eta_\pm(\la)(x-y)\sigma_3} \hat H^\pm(y,\la)\hat w^\pm(y,\la)
e^{-\ii\eta_\pm(\la)(x-y)\sigma_3}dy.
\end{equation}
The solutions $w^\pm(x,\la)$ are smooth in $t$ and $x$, related each other by the
formula:
$$
w^-(x,\la)=\sigma_2 w^{+*}(x,\la)\sigma_2.
$$
If $n(\la)<0$ then the first column $w^+[1](x,\la)$ has an analytic continuation $w[1](x,z)$ to the upper $\mathbb{C}_+$ complex half-plane and the second column $w^-[2](x,\la)$ has an analytic continuation $\overline w[2](x,z)$ to the lower $\mathbb{C}_-$ half-plane. The corresponding analytic vector-column
$w[1](x,z)e^{-\ii\eta(z) x}$ and $\overline w[2](x,z)e^{\ii\eta(z) x}$ are continuous in $\mathbb{C}_\pm\cup\mathbb{R}$ and have the following asymptotics:
\begin{itemize}
\item $w[1](x,z)e^{-\ii\eta(z) x}=\begin{pmatrix}
                                  1 \\
                                  0
\end{pmatrix} +O(z^{-1}), \qquad \Im\eta (z)\ge0, \quad z\to\infty;
$
\item $\overline w[2](x,z)e^{\ii\eta(z) x}=\begin{pmatrix}
                                  0 \\
                                  1
\end{pmatrix} +O(z^{-1}), \qquad \Im\eta(z)\le0, \quad z\to\infty.
$\end{itemize}
If $n(\la)>0$ then the first column $w^+[1](x,\la)$ and the second column $w^-[2](x,\la)$
have analytic continuations into domains $\mathbb{C}_+\setminus D^+$ and
$\mathbb{C}_-\setminus D^-$ respectively with the same asymptotic behavior, while
the first column $w^-[1](x,\la)$ and the second column $w^+[2](x,\la)$ have analytic (bounded) continuations $\overline w[1](x,z)$ and  $w[2](x,z)$ to the domains $D^-$ and $D^+$ respectively.
\end{lem}
The proof can be done by the same way as for Lemma 2.3.

Lemmas 2.4, 2.5 and formulas (\ref{Z}), (\ref{phi}), (\ref{psias1}), (\ref{psias2}), (\ref{w}), (\ref{hatw})
imply the following properties of the matrix $Z^\pm(t,x,\la)=$
$(Z^\pm[1](t,x,\la), Z^\pm[2](t,x,\la))$:
\begin{enumerate}
\item $Z^\pm(t,x,\la)$  satisfy the $t-$ and $x^\pm-$equations;
\item $\det Z^\pm(t,x,\la) \equiv1, \quad \la\in\mathbb{R} $;
\item the map $(x,t)\longmapsto Z^\pm(t,x,\la)$ is smooth in $t$
and $x$;
\item vector column $Z^+[1](t,x,\la)$ has an analytic continuation   $Z[1](t,x,z)$ for $z\in\mathbb{C}_+$ if $n(\la)<0$ ($z\in\mathbb{C}_+\setminus D^+$ if $n(\la)>0$),
    which is continuous in the closure of $\mathbb{C}_+$ ($\mathbb{C}_+\setminus D^+$)
    and
    $$
    Z[1](t,x,z)e^{\ii zt-\ii\eta(z)x}=  \begin{pmatrix}1\\0\end{pmatrix}+ O(z^{-1}), \quad  z\to\infty;
    $$
\item vector column $Z^-[2](t,x,\la)$ has an analytic continuation   $\overline Z[2](t,x,z)$ for $z\in\mathbb{C}_-$ if $n(\la)<0$ ($z\in\mathbb{C}_-\setminus D^-$ if $n(\la)>0$), which is continuous in the closure of $\mathbb{C}_-$ ($\mathbb{C}_-\setminus D^-$) and
    $$
    \overline Z[2](t,x,z)e^{-\ii zt+\ii\eta(z)x}=  \begin{pmatrix}0\\1\end{pmatrix}+ O(z^{-1}), \quad  z\to\infty.
    $$

If $n(\la)>0$ then

\item  vector column $Z^-[1](t,x,\la)$ has an analytic continuation   $\overline Z[1](t,x,z)$ to the domain $D^-$, and vector column $Z^+[2](t,x,\la)$ has an analytic continuation $Z[2](t,x,z)$  to the domain $D^+$, where the both analytic continuations are continuous up to the boundary of $D^+$ and $D^-$ respectively.
\end{enumerate}
To prove (4) and (5) we use asymptotic relations (\ref{psias1}), (\ref{psias2}) and
following one (where signs $\pm$ are omitted for a convenience)
$$
w(x,\la)\sim\begin{pmatrix}e^{\ii x\eta(\la)}&0\\0&e^{-\ii x\eta(\la)}\end{pmatrix}+
\begin{pmatrix}\chi_{11}(x,\la)e^{\ii x\eta(\la)}&\chi_{12}(x,\la)e^{-\ii x\eta(\la)}\\\chi_{21}(x,\la)e^{\ii x\eta(\la)}&\chi_{22}(x,\la)e^{-\ii x\eta(\la)}\end{pmatrix}.
$$
Here $\chi_{11}(x,\la), \,  \chi_{21}(x,\la) = \ord(\la^{-1})+\ord(\la^{-1}e^{2\ii L\la})$ and $\chi_{12}(x,\la), \,  \chi_{22}(x,\la) = \ord(\la^{-1})+\ord(\la^{-1}e^{-2\ii L\la})$ as $\la\to\pm \infty$ ($0\le x\le L<\infty$). If $L=\infty$  then $\chi_{ij}(x,\la)=\ord(\la^{-1})$ as $\la\to\pm \infty$. Indeed, for the first vector column, we have
$$
Z[1](t,x,z)e^{\ii zt-\ii\eta(z)x}\sim\left(\begin{pmatrix}
1&0\\0&e^{2\ii zt} \end{pmatrix}+O(z^{-1})\right)\left(
\begin{pmatrix}1\\0\end{pmatrix}+\begin{pmatrix}\chi_{11}(x,z)\\\chi_{21}(x,z)  \end{pmatrix} \right)=\begin{pmatrix}1\\0\end{pmatrix}+O(z^{-1}).
$$
For the second vector column there is
$$
\overline Z[2](t,x,z)e^{-\ii zt+\ii\eta(z)x}\sim\left[\begin{pmatrix}
e^{-2\ii zt} &0\\0&1\end{pmatrix}+O(z^{-1})\right]\left[
\begin{pmatrix}0\\1\end{pmatrix}+\begin{pmatrix}\chi_{12}(x,z)\\\chi_{22}(x,z)  \end{pmatrix} \right]=\begin{pmatrix}0\\1\end{pmatrix}+ O(z^{-1}).
$$

Due to analytic continuations of the first vector column  to the upper half plane and second vector column to the lower half plane we have $(\chi_{11}(x,z), \,  \chi_{21}(x,z)) = \ord(z^{-1}) $ for $\Im z\ge0$ and
$(\chi_{21}(x,z), \,  \chi_{22}(x,z)) = \ord(z^{-1}) $ for $\Im z\le0$ as $z\to\infty$.

Pairs of vectors ($\overline Y[1](t,x,z)$, $\overline Z[2](t,x,z)$) and
($Z[1](t,x,z)$, $Y[2](t,x,z)$) will alow to define below suitable matrix
Riemann-Hilbert problem in the case $n(\la)<0$. In the case $n(\la)>0$ we will use
additionally more vectors ($\overline Z[1](t,x,z)$, $\overline Y[1](t,x,z)$) and
($Y[2](t,x,z)$, $Z[2](t,x,z)$).

At the end of this section we consider the problem on  whole $t$-line studied in \cite{AKN} and \cite{GZM}, where the Marchenko integral equations were used. We would like to formulate below (in section 4) the corresponding matrix RH problem. Let  ${\mathcal E}(t,x)$ is vanishing as $t\to-\infty$ in such a way that there exists the Jost solution 
$$
\hat\Psi(t,x,\la)=e^{-\ii\la t\sigma_3}+\int\limits_{-\infty}^t \hat L(t,\tau,x)
e^{-\ii\la\tau\sigma_3}d\tau
$$
of $t$-equation (\ref{teq}). Let ${\mathcal N}(t,x,\la)-\sigma_3$ and $\rho(t,x,\la)$ are also vanishing as $t\to-\infty$. In this case we use the following pair of compatible solutions:
$Y^\pm(t,x,\la)$ are the same as above, while
$$
Z^\pm(t,x,\la)=\hat\Psi(t,x,\la)e^{\ii x\eta_\pm(\la)\sigma_3}.
$$
Due to the Lemma 2.1 $Z^\pm(t,x,\la$) is compatible solution of (\ref{teq}) and
(\ref{pmxeq}) because, when ${\mathcal N}(t,x,\la)=\sigma_3$ and $\rho(t,x,\la)=0$,
$x^\pm$-equation takes the form $E_x=\ii\eta_\pm(\la)\sigma_3 E$ and we put
$E^\pm(x,\la)=e^{\ii x\eta_\pm(\la)\sigma_3}$. Further, vector column $Z^+[1](t,x,\la)$
has an analytic continuation $Z[1](t,x,z)$ for $z\in\mathbb{C}_+$,
and vector column $Z^-[2](t,x,\la)$ has an analytic continuation
$\overline Z[2](t,x,z)$ for $z\in\mathbb{C}_-$. They satisfy the properties (4)
and (5) above for any choice of the sign of the function $n(\la)$.
\section{Transition matrices and spectral functions}
\setcounter{equation}{0}

Now we consider a dependence between solutions of the AKNS equations, corresponding transition matrices and spectral functions. First of all, since
$ Y^\pm(t,x,\la)$ and    $Z^\pm(t,x,\la)$  are solutions
of the $t-$ and $x^\pm-$equations (\ref{teq}), (\ref{pmxeq}), they
are linear dependent. So there exist
transition matrices $T^\pm(\la)$, independent of $x$ and $t$, such that
\begin{equation}
\label{sc}Y^\pm(t,x,\la)=Z^\pm(t,x,\la) T^\pm(\la), \qquad \Im\la=0.
\end{equation}
They are equal to
\[
T^\pm(\la)= (w^\pm(0,\la))^{-1}\Phi(0,\la), \qquad \Im\la=0
\]
and have the following structure:
\[
T^\pm(\la)= \begin{pmatrix}\overline a^\pm(\la)&b^\pm(\la)\\
-\overline b^\pm(\la)&a^\pm(\la)\end{pmatrix}.
\]

Matrices $w^\pm(0,\la)$ have the similar structure
$$
w^\pm(0,\la)=\hat w^\pm(0,\la)=\begin{pmatrix}
\alpha^\pm(\la)&-\overline\beta^\pm(\la)\\\beta^\pm(\la)&\overline\alpha^\pm(\la)\end{pmatrix}
$$
and define the spectral functions of $x^\pm$-equation. Since coefficients of $x^\pm$-equation are taken for the fixed time $t=0$, the spectral functions are
uniquely defined by given initial functions ${\mathcal E}(0,x)$,
$\rho(0,x,\la)$ and ${\mathcal N}(0,x,\la)$. They are not independent because the solutions $w^\pm(x,\la)$ of the $x^\pm$-equations
are related each other by $w^-(x,\la)=\sigma_2 w^{+*}(x,\la)\sigma_2,$
that yields the following reductions:
\begin{align*}
\overline\alpha^\pm(\la)=&\alpha^{\mp\,*}(\la),\\
\overline\beta^\pm(\la)=&-\beta^{\mp\,*}(\la).
\end{align*}
The matrix
$$\Phi(0,\la)=\begin{pmatrix}
\overline A(\la)&B(\la)\\-\overline B(\la)&A(\la)\end{pmatrix}$$
gives the spectral functions of the $t$-equation with $x=0$. They are uniquely
defined by the boundary condition ${\mathcal E}(t,0)$. Analogous reduction condition
$\Phi(t,\la)=\sigma_2 \Phi^{+*}(t,\la)\sigma_2$ means
$$
\overline A(\la)=A^*(\la),\quad \overline B(\la)=B^*(\la).
$$
Finally we have that the spectral functions, generated by the transition matrix $T^\pm(\la)$,
are defined by the initial and boundary conditions. This matrix satisfies the analogous
reduction:
$$
T^-(\la)=\sigma_2 T^{+*}(\la)\sigma_2,
$$
which is equivalent to the following ones:\begin{align}\label{ab}
\overline a^{\,\pm}(\la)=&a^{\mp\,*}(\la),\nonumber\\
\overline b^{\,\pm}(\la)=&b^{\mp\,*}(\la).
\end{align}
Due to A.S.Fokas \cite{F1} we call all of these functions as the spectral functions.

Entries of the matrix $T^\pm(\la)$ are
\begin{align*}
\overline a^{\,\pm}(\la)=&\det[Y^\pm[1](t,x,\la),Z^\pm[2](t,x,\la)] \\
\overline b^{\,\pm}(\la)=&\det[Y^\pm[1](t,x,\la),Z^\pm[1](t,x,\la)]\\
b^\pm(\la)=&\det[Y^\pm[2](t,x,\la),Z^\pm[2](t,x,\la)] \\
a^\pm(\la)=&\det[Z^\pm[1](t,x,\la),Y^\pm[2](t,x,\la)].
\end{align*}
If $n(\la)<0$ these relations show that  $a^+(\la)$ has an analytic continuation $a(z)$ for $z\in\mathbb{C}_+$, and  $a^-(\la)$ has an analytic continuation
$\overline a(z)$ for $z\in\mathbb{C}_-$ ($\overline a(z)=a^*(z^*)$).
Functions $b^\pm(\la)$ and $\overline b^{\,\pm}(\la)$ are defined for
$\la\in\mathbb{R}$ only. Determinant of $T^\pm(\la)\equiv1$ for $\Im\la=0$. The
spectral functions have the following asymptotics:
\begin{align*}
\overline a(z)=&1+O(z^{-1}), \qquad \Im z\le0, \qquad   |z|\to\infty.\\
b^\pm(\la)=&O(\la^{-1}),\qquad\qquad\qquad\qquad\qquad\la\to\pm\infty,\\
a(z)=&1+O(z^{-1}),\qquad \Im z\ge0, \qquad   |z|\to\infty.
\end{align*}

If   $a(z)$ have zeroes $z_j\in\mathbb{C}_+$ ($j=1,2,....,p$) then
$$a(z_j)=\det[Z[1](t,x,z_j),Y[2](t,x,z_j)]=0.$$ Hence vector-columns of the determinant are linear dependent:
\be\label{gj1}
Y[2](t,x,z_j)=\gamma_jZ[1](t,x,z_j), \qquad \gamma_j=\D \frac{B(z_j)}{\alpha(z_j)}=
\D \frac{A(z_j)}{\beta(z_j)}.
\ee
At the conjugated points $z^*_j\in\mathbb{C}_-$ ($j=1,2,....,p$)
$$\overline a(z^*_j)=\det[\overline Y[1](t,x,z^*_j),\overline Z[2](t,x,z^*_j)]=0.$$ Therefore
\be\label{gj2}
\overline Y[1](t,x,z^*_j)=\overline\gamma_j\overline Z[2](t,x,z^*_j), \qquad \overline\gamma_j=\D \frac{B^*(z^*_j)}
{ \alpha^*(z^*_j)}=\D \frac{A^*(z^*_j)}{\beta^*(z^*_j)}=\gamma^*_j.
\ee

If $n(\la)>0$ we have another analytic properties. Namely,
$a^+(\la)$ has an analytic continuation $a(z)$ for
$z\in\mathbb{C}_+\setminus D^+$, and  $a^-(\la)$ has an analytic continuation
$\overline a(z)$ for $z\in\mathbb{C}_-\setminus D^-$ ($\overline a(z)=a^*(z^*)$).
The function $\overline b^{\,-}(\la)$ has an analytic continuation $\overline b(z)$ for $z\in D^-$, and  $b^+(\la)$ has an analytic continuation
$b(z)$ for $z\in D^+$ (and $ b^*(z^*)=\overline b(z)$).
The spectral functions have the same asymptotics as above.

If   $a(z)$ have zeroes $z_j\in\mathbb{C}_+\setminus D^+$ ($j=1,2,....,p$) then
$$a(z_j)=\det[Z[1](t,x,z_j),Y[2](t,x,z_j)]=0.$$ Hence vector-columns of the determinant are linear dependent:
\be\label{gj3}
Y[2](t,x,z_j)=\gamma_jZ[1](t,x,z_j), \qquad \gamma_j=\D \frac{B(z_j)}{\alpha(z_j)}=
\D \frac{A(z_j)}{\beta(z_j)}.
\ee
At the conjugated points $z^*_j\in\mathbb{C}_-\setminus D^-$ ($j=1,2,....,p$)
$$\overline a(z^*_j)=\det[\overline Y[1](t,x,z^*_j),\overline Z[2](t,x,z^*_j)]=0.$$ Therefore
\be\label{gj4}
\overline Y[1](t,x,z^*_j)=\overline\gamma_j\overline Z[2](t,x,z^*_j), \qquad \overline\gamma_j=\D \frac{B^*(z^*_j)}
{ \alpha^*(z^*_j)}=\D \frac{A^*(z^*_j)}{\beta^*(z^*_j)}=\gamma^*_j.
\ee
Being analytic in the domain $D_+$ the function $b(z)$ can have zeroes at the points $\hat z_j\in D_+$ ($j=1,2,....,q$). Hence
\be\label{gj5}
Y[2](t,x,\hat z_j)=\hat \gamma_jZ[2](t,x,\hat z_j), \qquad \hat\gamma_j=
\D \frac{B(\hat z_j)}{\alpha(\hat z_j)}=
\D \frac{A(\hat z_j)}{\beta(\hat z_j)}.
\ee
At the conjugated points $\hat z^*_j\in D^-$ ($j=1,2,....,q$)
\be\label{gj6}
\overline Y[1](t,x,\hat z^*_j)=\hat{\overline\gamma_j}\overline Z[1](t,x,\hat z^*_j),
\qquad \hat{\overline\gamma_j}=\D \frac{B^*(\hat z^*_j)}
{ \alpha^*(\hat z^*_j)}=\D \frac{A^*(\hat z^*_j)}{\beta^*(\hat z^*_j)}=\gamma^*_j.
\ee
\section{ Matrix Riemann-Hilbert problems}
\setcounter{equation}{0}

In this section we give a reconstruction of the solution of the
mixed problem to the MB equations in terms of  the  associated matrix
Riemann-Hilbert problem.

In the case $n(\la)<0$ we put
\begin{align} \nonumber
   \hat\Phi(t,x,z)=\begin{pmatrix}
     Z[1](t,x,z) &\D \frac{Y[2](t,x,z)}{a(z)}
   \end{pmatrix}, & \quad z\in\mathbb{C}_+,\\\label{hatPhi}\\
   =\begin{pmatrix}
     \D \frac{\overline Y[1](t,x,z)}{{\overline a}(z)}&\overline Z[2](t,x,z)
   \end{pmatrix},& \quad z\in\mathbb{C}_-.\nonumber
\end{align}
Scattering relations (\ref{sc}) gives:  \[
\D \frac{Y^-[1](t,x,\la)}{\overline a^{\,
-}(\la)}=Z^-[1](t,x,\la)- \overline r^{\, -}(\la)Z^-[2](t,x,\la),
\]
\[
\D \frac{Y^+[2](t,x,\la)}{a^+(\la)}=r^+(\la)Z^+[1](t,x,\la)+
Z^+[2](t,x,\la),
\]
where
$$
\overline r^{\, -}(\la)=\frac{\overline  b^{\, -}(\la)}{\overline a^{\,-}(\la)}=
\D\frac{\alpha^-(\la)\overline B(\la)+\beta^+(\la)\overline A(\la)}
{\overline{\alpha}^-(\la)\overline A(\la)-\overline{\beta}^-(\la)\overline B(\la)},
\qquad r^+(\la)=\frac{b^+(\la)}{a^+(\la)}=
\D\frac{\overline{\alpha}^+(\la)B(\la)+\overline{\beta}^+(\la)A(\la)}
{\alpha^+(\la)A(\la)-\beta^+(\la)B(\la)}.
$$
Using these relations we find that
$$\det\hat\Phi(t,x,\la+\ii0)=\det\hat\Phi(t,x,\la-\ii0)\equiv1, \qquad \Im\la=0.$$
Since $\hat\Phi_+(t,x,\la):=\hat\Phi(t,x,\la+\ii0)$ and
$\hat\Phi_-(t,x,\la):=\hat\Phi(t,x,\la-\ii0)$ satisfy $t$-equation, we
have
$$
\hat\Phi_-(t,x,\la)=\hat\Phi_+(t,x,\la)J_0(x,\la),
$$
where unimodular matrix
$$
J_0(x,\la)=\begin{pmatrix}
     Z^+[1](0,x,\la) &\D \frac{Y^+[2](0,x,\la)}{a^+(\la)}
   \end{pmatrix}^{-1}\begin{pmatrix}
     \D \frac{Y^-[1](0,x,\la)}{\overline a^{\, -}(\la)}&Z^-[2](0,x,\la)
   \end{pmatrix}.
$$
All entries
\begin{align*}
Z^+[1](0,x,\la)&=W^+(0,x,\la)w^+[1](0,\la), \\
Y^+[2](0,x,\la)&=W^+(0,x,\la)\Phi[2](0,\la),\\
Y^-[1](0,x,\la)&=W^-(0,x,\la)\Phi[1](0,\la),\\
Z^-[2](0,x,\la)&=W^-(0,x,\la)w^-[2](0,\la).
\end{align*}
of these matrices are known by initial and boundary conditions. Indeed, since $W^\pm(0,x,\la)=w^\pm(x,\la)(w^\pm(0,\la))^{-1}$ then
\begin{align*}
Z^+[1](0,x,\la)&= w^+(x,\la) \begin{pmatrix}
                               1 \\
                               0
                             \end{pmatrix}, \\
Y^+[2](0,x,\la)&=w^+(x,\la) \begin{pmatrix}
                               b^+(\la) \\
                               a^+(\la)
                             \end{pmatrix}, \\
Y^-[1](0,x,\la)&=w^-(x,\la) \begin{pmatrix}
                               \overline{a}^-(\la)\\
                               \overline{b}^-(\la)
                             \end{pmatrix}, \\
Z^-[2](0,x,\la)&=w^-(x,\la) \begin{pmatrix}
                               0 \\
                               1
                             \end{pmatrix}.
\end{align*}
The matrix $ J_0(x,\la)$ can be written in the form:
\be\label{J0} J_0(x,\la)=(S^+(\la))^{-1}(w^+(x,\la))^{-1}w^-(x,\la)S^-(\la)\ee
where
\begin{align*} S^+(\la)=&\begin{pmatrix}
                   1 & r^+(\la) \\
                   0 & 1 \\
                 \end{pmatrix},\qquad r^+(\la)=\D\frac{b^+(\la)}{a^+(\la)},\\
   S^-(\la)=&\begin{pmatrix}
                   1 & 0\\
                   -\overline {r}^-(\la)&1
                 \end{pmatrix},\qquad r^-(\la)=\D \frac{\overline{b}^-\la)}{\overline a^{\, -}(\la)}
\end{align*}
are spectral matrices defined by initial ${\mathcal E}_{in}(t)$ and
boundary  ${\mathcal E}_{0}(x)$, $\rho_0(x,\la)$, ${\mathcal N}_0(x,\la)$
conditions. They are unimodular matrices. Matrix $w^\pm(x,\la)$ is the Jost solution of $x^\pm$-equation (with $t=0$):
\begin{align}\label{xpm0}
   \frac{d}{dx}w^\pm(x,\la)=&(\ii\la\sigma_3-\ii
   G_\pm(0,x,\la)+H(0,x))w^\pm(x,\la),\\
    &w^\pm(L,\la)e^{-\ii L\eta_\pm(\la)\sigma_3}=I. \nonumber
\end{align}
It is defined by given boundary conditions only. The jump matrix $J_0(x,\la)$
can be rewritten in the form:
$$
J_0(x,\la)=(K^+(x,\la))^{-1}K^-(x,\la),
$$
where $K^\pm(x,\la)$ satisfies equation (\ref{xpm0}) under
condition:
$$
 K^\pm(L,\la)=e^{\ii L\eta_\pm\sigma_3}S^\pm(\la).
$$
The following reductions
\begin{align*}
S^-(\la)=&\sigma_2 S^{+*}(\la)\sigma_2,\\
w^-(x,\la)=&\sigma_2 w^{+*}(x,\la)\sigma_2\\
K^-(x,\la)=& \sigma_2 K^{+\, *}(x,\la)\sigma_2
\end{align*}
are fulfilled. These reductions are equivalent to the following
ones:
\begin{align*}[S^+(\la)]^{-1}=&[S^-(\la)]^\dag,\\
[w^+(x,\la)]^{-1}=&[w^-(x,\la)]^\dag,\\
[K^+(x,\la)]^{-1}=&[K^-(x,\la)]^\dag
\end{align*}
where $\dag$ means Hermitian conjugation.
\begin{lem}
For any $x\in\mathbb{R}_+$ and $\la\in\mathbb{R}$ the jump matrix
$$
J_0(x,\la)=(K^-(x,\la))^{\dag}K^-(x,\la)
$$
is positive defined.
\end{lem}
\begin{proof} Indeed, scalar product
$$(J_0(x,\la)\xi, \xi)=(K^-(x,\la)\xi, K^-(x,\la)\xi)$$
is positive for any $\xi\in\mathbb{C}^2$, $\xi\neq0$. Suppose the contrary, i.e.  there exists $\xi_0=\xi_0(x,\la)\neq0$ such that $K^-(x,\la)\xi_0(x,\la)=0$. Since $\det K^-(x,\la)\equiv1$,  vector function $\xi_0(x,\la)$ is equal zero identically. Hence $(J_0(x,\la)\xi, \xi)>0$.
\end{proof}
\begin{lem}
For any fixed $x\in\mathbb{R}_+$ the jump matrix $J_0(x,\la)$ has the following asymptotic behavior:
$$
J_0(x,\la)=I+\ord(\la^{-1}), \qquad \la\to\pm\infty.
$$
\end{lem}
\begin{proof}
The matrix $K:=K^-(x,\la)$ satisfies equation:
$$
K^\prime_x= (\ii\la\sigma_3 +H(0,x)-\ii G_-(0,x,\la))K
$$
Hermitian conjugation matrix $K^\dag$ satisfies the following one:
$$
K^{\dag\prime}_x= K^\dag(-\ii\la\sigma_3 +H^\dag(0,x)+\ii G_-^{\dag}(0,x,\la))
$$
Since $H^\dag(0,x)=-H(0,x)$ and $G_-^{\dag}(0,x\la)=G_{+}(0,x,\la)$, the last equation takes the form:
$$
K^{\dag\prime}_x= K^\dag(-\ii\la\sigma_3 -H^\dag(0,x)+\ii G_{+}(0,x,\la)).
$$
Take into account definition of $J_0(x,\la)$ we find
$$
J_{0x}^\prime=\ii K^\dag(G_{+}(0,x,\la)-G_{-}(0,x,\la))K=-\frac{\pi n(\la)}{2} K^\dag F(0,x,\la)K.
$$
Since $F(0,x,\la)$ and $K=K^-(x,\la)$ are bounded, and  $n(\la)=o(\la^{-1})$ as
 $\la\to\pm\infty$, we have that $J_0(x,\la)=\hat J(\la)+ o(\la^{-1})$ as $\la\to\pm\infty$. Then $\hat J(\la)=J_0(0,\la)+ o(\la^{-1})$ and $J_0(0,\la)=K^{+-1}(0,\la)K^-(0,\la)=I+\ord(\la^{-1})$ as $\la\to\pm\infty$. Hence $\hat J(\la)=I+\ord(\la^{-1})$ as $\la\to\pm\infty$.
\end{proof}

For the problem on whole $t$-line the jump matrix $J_0(x,\la)$ can be found explicitly. Indeed, for $t\to-\infty$ we have
\begin{align}\label{Jline}
J_0(x,\la)=&\begin{pmatrix}
     Z^+[1](t,x,\la) &\D \frac{Y^+[2](t,x,\la)}{a^+(\la)}
   \end{pmatrix}^{-1}\times\begin{pmatrix}
     \D \frac{Y^-[1](t,x,\la)}{\overline a^{\, -}(\la)}&Z^-[2](t,x,\la)
   \end{pmatrix}\nonumber\\=&\begin{pmatrix}
                   e^{\ii\la t-\ii x\eta_+(\la)} &-\D\frac{b^+(\la)}{a^+(\la)} e^{-\ii\la t+\ii x\eta_+(\la)} \\\\
                   0 & e^{-\ii\la t+\ii x\eta_+(\la)} \\
                 \end{pmatrix}
\begin{pmatrix}e^{-\ii\la t+\ii x\eta_-(\la)} & 0 \\\\
-\D\frac{\overline b^{\, -}(\la)}{\overline a^{\, -}(\la)}e^{\ii\la t-\ii x\eta_-(\la)} & e^{\ii\la t-\ii x\eta_-(\la)}\end{pmatrix} \\
=&\begin{pmatrix}
e^{-\ii x(\eta_+(\la)-\eta_-(\la))}+\overline r^{\, -}(\la)r^+(\la)e^{\ii x(\eta_+(\la)-\eta_-(\la))} & -r^+(\la)e^{\ii x(\eta_+(\la)-\eta_-(\la))}\\\\
-\overline r^{\, -}(\la)e^{\ii x(\eta_+(\la)-\eta_-(\la))}&e^{\ii x(\eta_+(\la)-\eta_-(\la))}\\\end{pmatrix},\nonumber
\end{align}
where $\overline r^{\, -}(\la)=r^{+\, *}(\la)$.

Now we put
\be\label{defM} M(t,x,z)=\hat\Phi(t,x,z)e^{\ii(zt-x\eta(z))\sigma_3},\qquad z\in\mathbb{C}\setminus\mathbb{R},
\ee
where $\eta(z)$ is defined by (\ref{eta}).
Then $M(t,x,z)$ is a solution of the following
Riemann-Hilbert problem $RH_{tx}$:
{\it \begin{itemize}
\item $M(t,x,z)$ is analytic ($a(z)\neq0$) or
meromorphic ($a(z)$ has finite number of zeroes) in
$z\in\mathbb{C}\setminus\mathbb{R}$ and continuous up to the real $\la$-axis;
\hfill{$RH1$}
\item If $a(z_j)=\overline a(z_j^*)=0, j=1,2,...,p$ then
$M(t,x,z)$ has poles at points $z=z_j$, $z=z^*_j$ ($j=1,2,...,p$),
and corresponding residues satisfy relations:\\

$\res\limits_{z=z_j} M[2](t,x,z)=m_je^{-2\ii(z_jt-x\eta(z_j))} M[1](t,x,z_j)$ \hfill{$RH2$}\\

$\res\limits_{z=z^*_j} M[1](t,x,z)=m^*_je^{-2\ii(z_jt-x\eta(z^*_j))}
M[2](t,x,z^*_j),$ \hfill{$RH3$}\\\\
where $m_j=\gamma_j/\dot a(z_j)$, $m^*_j=\overline\gamma_j/\dot{\overline a}^{\, -}
(z^*_j)$, and numbers $\gamma_j$, $\overline\gamma_j=\gamma^*_j$ are defined in
(\ref{gj1}) and (\ref{gj2});
\item  $M_{-}(t,x,\la)=M_{+}(t,x,\la)J(t,x,\la), \quad
\la\in\mathbb{R},$\hfill{$RH4$}
\be
J(t,x,\la)= e^{-\ii(\la t-x\eta_+(\la))\sigma_3}J_0(x,\la)e^{\ii(\la t-x\eta_-(\la))\sigma_3}; \label{J1}
\ee
\item $ M(t,x,z)=I+O(z^{-1}),\quad |z|\to\infty.$ \hfill{$RH5$}
\end{itemize}}
Analytical properties follow from (\ref{hatPhi}), and residues relations arise from (\ref{gj1})-(\ref{gj2}). The jump matrix $J_0(x,\la)$ in (\ref{J1}) was described in (\ref{J0}), (\ref{xpm0}) and Lemmas 4.1 and 4.2. The asymptotic behavior of $ M(t,x,z)$ follows from sections 2 and 3.

For the problem on whole $t$-line the jump matrix degenerates to explicit one:
\begin{align*}
J(t,x,\la)=&e^{-\ii(\la t-x\eta_+(\la))\sigma_3}J_0(x,\la)e^{\ii(\la t-x\eta_-(\la))\sigma_3}\\\\
=&\begin{pmatrix}
    1+|r^+(\la)|^2e^{2\ii x(\eta_+(\la)-\eta_-(\la))} &
    -r^+(\la)e^{-2\ii\la t+2\ii x\eta_+(\la)}\\  \\
    -r^{+*}(\la)e^{2\ii\la t-2\ii x\eta_-(\la)}  & 1 \\
\end{pmatrix}.
\end{align*}
Here we have used (\ref{Jline}), (\ref{ab}).
Since $\eta_\pm=\eta(\la\pm\ii0)=\la-I_0(\la)\mp\D\frac{\pi\ii}{4}n(\la)$ where
$I_0(\la)={\rm p.v.}\int\limits_{-\infty}^\infty\D\frac{n(s)}{s-\la}\D\frac{ds}{4}$ then
$$
J(t,x,\la)=I+\ord(e^{\pi n(\la)x/2})\rightarrow I, \qquad x\to+\infty
$$
for a long attenuator ($n(\la)<0$), but it is unbounded ($\ord(e^{\pi n(\la)x})$) for a long amplifier ($n(\la)>0$).  Therefore the jump matrix is exponentially closed to the identity matrix if $n(\la)<0$. Hence the main term of asymptotics contains solitons generated by discrete spectrum, while a contribution of continuous spectrum is exponentially small. This justifies the phenomenon of self-induced transparency of attenuators. For the first time it was proved in \cite{AKN}.

In order to define a matrix RH problem for the mixed problem, we introduce another matrix
\begin{align*}
   \hat\Phi(t,x,z)=\begin{pmatrix}
     Z[1](t,x,z) &\D \frac{Y[2](t,x,z)}{a(z)}
   \end{pmatrix}, & \quad z\in\mathbb{C}_+\setminus D^+,\\
   =\begin{pmatrix}
     \D \frac{\overline Y[1](t,x,z)}{{\overline a}(z)}&\overline Z[2](t,x,z)
   \end{pmatrix},& \quad z\in\mathbb{C}_-\setminus D^-,\\
   =\begin{pmatrix}
     \D \frac{Y[2](t,x,z)}{b(z)}&Z[2](t,x,z)
   \end{pmatrix}, & \quad z\in D^+,\\
   =\begin{pmatrix}
     \overline Z[1](t,x,z)&\D \frac{\overline Y[1](t,x,z)}{{\overline b}(z)}
   \end{pmatrix},& \quad z\in  D^-,
\end{align*}
when $n(\la)>0$.
By the same way, using residues relations (\ref{gj3})-(\ref{gj6}), we obtain a  meromorphic matrix RH problem. But for the sake of simplicity we assume that $a(z)\neq0$ in the domain $\mathbb{C}_+\setminus D^+$ and $b(z)\neq0$ in the domain $D^+$, i.e. we consider a regular RH problem. In this case matrix
$M(t,x,z)=\hat\Phi(t,x,z)e^{\ii(zt-x\eta(z))\sigma_3}$ satisfies the next items.
{\it\begin{itemize}
\item $M(t,x,z)$ is analytic
$z\in\mathbb{C}\setminus\Sigma, \, \Sigma:= \mathbb{R}\cup\gamma\cup\bar\gamma$ and continuous up to the contour $\Sigma$;
\hfill{$RRH1$}
\item  $M_{-}(t,x,z)=M_{+}(t,x,z)J(t,x,z), \quad
z\in\Sigma,$\hfill{$RRH2$}
$$
J(t,x,z)= e^{-\ii(\la t-x\eta_+(\la))\sigma_3}J_0(x,\la)e^{\ii(\la t-x\eta_-(\la))\sigma_3},\qquad \la\in\mathbb{R}\\
$$$$\hskip3cm=e^{-\ii(zt-x\eta(z))\sigma_3}J_0(z)e^{\ii(zt-x\eta(z))\sigma_3},
\qquad z\in\gamma\cup\bar\gamma, \label{J}
$$
where
\begin{align*}
J_0(x,z)=&(K_0^+(x,\la))^{-1}K_0^-(x,\la),& z=\la\in\mathbb{R}_D=\mathbb{R}\cap\overline{D_-\cup D_+},\\\\
=&(K^+(x,\la))^{-1}K^-(x,\la),& z=\la\in\mathbb{R}\setminus\mathbb{R}_D.
\end{align*}
\begin{align*}
J_0(z)=& \begin{pmatrix} 0&-\D\frac{b(z)}{a(z)} \\
\D\frac{a(z)}{b(z)} &1\end{pmatrix},& z\in\gamma,\\
=& \begin{pmatrix} 1&-\D\frac{\bar a(z)}{\bar b(z)} \\
\D\frac{\bar b(z)}{\bar a(z)} &0\end{pmatrix},& z\in\bar\gamma
\end{align*}
\item $\det J(t,x,z)\equiv1$  for $z\in\Sigma$;
\item $ M(t,x,z)=I+O(z^{-1}),\quad |z|\to\infty.$ \hfill{$RRH3$}
\end{itemize}
Here $K_0^\pm(x,\la)=w^\pm(x,\la)S_0^\pm(\la),\qquad K^\pm(x,\la)=w^\pm(x,\la)S^\pm(\la)$ and
\begin{align*}
S_0^+(\la)=&\begin{pmatrix}1&0\\\\\D\frac{a^+(\la)}{b^+(\la)}&1\end{pmatrix}&
S_0^-(\la)=&\begin{pmatrix}1&-\D\frac{\overline a^-(\la)}{\overline b^-(\la)}\\\\0&1\end{pmatrix},\\
S^+(\la)=&\begin{pmatrix}1&\D\frac{b^+(\la)}{a^+(\la)}\\\\0&1\end{pmatrix},&
S^-(\la)=&\begin{pmatrix}1&0\\\\-\D\frac{\overline b^-(\la)}{\overline a^-(\la)}&1\end{pmatrix}.
\end{align*}}

Note that matrix $J_0(x,z)$ depends on $x$ for $z=\la\in \mathbb{R}$,  and matrix $J_0(z)$ is independent on $x$ (and on $t$ as well) for $z\in\gamma\subset\mathbb{C}_+$ and $z\in\bar\gamma\subset\mathbb{C}_-$. Thus we prove
\begin{thm}
Let ${\mathcal E}(t,x)$, ${\mathcal N}(t,x,\lambda)$ and $\rho(t,x,\lambda)$ be the solution of the mixed problem (\ref{ic})-(\ref{bc}) to  the Maxwell-Bloch equations (\ref{MB1})-(\ref{MB3}). Then there exists matrix $M(t,x,z)$ which is the solution of the meromorphic Riemann-Hilbert problem ($RH1$)-($RH5$) if $n(\la)<0$ or regular (under additional conditions $a(z)\neq0$ and $b(z)\neq0$) RH problem ($RRH1$)-($RRH3$) if $n(\la)>0$.  Complex electric field envelope ${\mathcal E}(t,x)$ is defined by relation:
\begin{align}\label{El}
{\mathcal E}(t,x)=&-\lim\limits_{z\to\infty}4\ii z M_{12}(t,x,z),
\end{align}
and ${\mathcal N}(t,x,\lambda)$ and $\rho(t,x,\lambda)$ can be found from linear equations (\ref{MB2})-(\ref{MB3}) by already known ${\mathcal E}(t,x)$.
\end{thm}
\begin{proof}
Formula (\ref{El}) follows from (\ref{teq}) and ($RH5$). Indeed, substituting (\ref{defM}) into
equation (\ref{teq}), we find
\be\label{Mt}
M_t+\ii z [\sigma_3,M]+HM=0.
\ee
Using ($RH5$) we put
$$
M(t,x,z) = I+\D \frac{m(t,x)}{z}+\mathrm{o}(z^{-1}),
$$
where
$$
m(t,x)=\lim\limits_{z\to\infty} z(M(t,x,z)-I).
$$
This asymptotics and equation (\ref{Mt}) give
\be\nonumber
H(t,x)=-\ii[\sigma_3,m(t,x)]
\ee
and hence
$$
{\mathcal E}(t,x)=-4\ii m_{12}=-\lim\limits_{z\to\infty}4\ii z M_{12}(t,x,z).
$$
\end{proof}
\emph{ Thus the mixed problem on the finite interval $0<x<L$ (or  half-line $0<x<\infty$) for the Maxwell-Bloch equations is linearizable completely.
}
\section{More general matrix Riemann-Hilbert problem}
\setcounter{equation}{0}

Now we prove that any Riemann-Hilbert problem like $RRH1$-$RRH3$ generates a solution
to the Maxwell-Bloch equations. From here and below we will consider more general construction.
Let oriented contour $\Sigma$ contains real line $\mathbb{R}$, circle
$\Gamma$ of sufficiently large radius and some finite arcs $\gamma_j\cup\overline\gamma_j$ ($j=1,2,....,p$), which
are symmetric over the real line. Thus
$$
\Sigma=\mathbb{R}\cup\Gamma\cup\bigcup\limits_{j=1}^p\gamma_j\cup\overline\gamma_j.
$$
If $n(\la)>0$ then contour $\Sigma$ includes additionally the closed oval $\gamma\cup\bar\gamma$ where $\Im\eta(z)=0$.
Such contours appear when we deal with periodic initial data or/and
periodic boundary conditions. The large circle $\Gamma$ allows to bypass   difficulties connected with discrete spectrum and spectral singularities (sf.\cite{Zhou}, \cite{MK}). Contour $\Sigma $ has the following orientation: real line $\mathbb{R}$ is oriented from the left to the right, the circle $\Gamma$ and oval $\gamma\cup\bar\gamma$  are oriented clock-wise, the arcs $\gamma_j\cup\overline\gamma_j$ are oriented updown. Then the formulation of a regular matrix RH problem is as follows.

{\it Find $2\times2$ matrix $M(t,x,z)$ such that
\begin{itemize}
\item $M(t,x,z)$ is analytic in
$z\in\mathbb{C}\setminus\Sigma$ and bounded up to the contour $\Sigma$;
\hfill{$R1$}
\item  $M_{-}(t,x,z)=M_{+}(t,x,z)J(t,x,z), \quad
z\in\Sigma;$\hfill{$R2$}
\item $\det J(t,x,z)\equiv1$  for $z\in\Sigma$;\hfill{$R3$}
\item $ M(t,x,z)=I+O(z^{-1}),\quad |z|\to\infty.$ \hfill{$R4$}
\end{itemize}
}

Let contour $\Sigma$ and jump matrix $J(t,x,z)$ are satisfied the Schwartz reflection principle:
\begin{itemize}
\item  contour $\Sigma$ is symmetric over real axis $\mathbb{R}$,
\item $J^{-1}(x,t,z)=  J^\dag(x,t,z^*)$ for $z\in\Sigma$ and $\Im z\neq0$,
\end{itemize}
where $\dag$ and $*$ are Hermitian and complex conjugations respectively.

Furthermore,
\begin{itemize}
\item  jump matrix $J(t,x,\la)$ for $\la\in\mathbb{R}$ and $x\in\mathbb{R}$ has a positive definite
real part and following asymptotic behavior:
$$
J(t,x,\la)=I+\ord(\la^{-1}), \qquad \la\to\pm\infty.
$$
\end{itemize}
\begin{thm}\label{t4.2}
Let jump matrix $J(t,x,z)$ satisfies the Schwartz reflection principle, has a positive definite real part and  $I-J(t,x,.)\in L^2(\Sigma)\cap L^\infty(\Sigma)$. Then for any fixed $t, x\in \mathbb{R}$, the regular RH problem $R1$, $R2$, $R3$, $R4$ has a unique solution $M(t,x,z)$.
\end{thm}

\begin{proof}
\emph{ Existence.} Let $t$ and $x$ be fixed.
We look for the solution  $ M(t,x,z)$ of the RH problem in the form
\begin{equation}\label{M}
M(t,x,z) =I+ \frac{1}{2\pi\ii}\int\limits_\Sigma
 \frac{P(t,x,s)[I- J(t,x,s)]}{s-z}ds, \qquad
z\notin\Sigma.
\end{equation}
The Cauchy integral (\ref{M}) provides all properties
of the RH problem (cf.\cite{D}) if and only if the matrix $Q(t,x,\la):=P(t,x,\la)-I$ satisfies the
singular integral equation
\begin{equation}\label{K}
Q(t,x,z)-{\mathcal K}[Q](t,x,z) =R(t,x,z), \qquad z\in\Sigma.
\end{equation}
The singular integral operator ${\mathcal K}$ and the right hand side $R(t,x,z)$ are as follows:
\begin{eqnarray*}{\mathcal K}[Q](t,x,z):=\D  \frac{1}{2\pi\ii}
\int\limits_\Sigma\D   {\frac{
Q(t,x,s)[I-J(t,x,s)]}{s-z_+}ds},\\\nonumber
R(t,x,z):=\D\frac{1}{2\pi\ii}\int\limits_\Sigma \D{\frac{
I-J(t,x,s)}{s-z_+}ds}.
\end{eqnarray*}
We consider this integral equation in the space $L^2(\Sigma)$ $2\times2$ matrix
complex valued functions $Q(z):=Q(t,x,z)$, $z\in\Sigma$. The norm of
$Q\in L^2(\Sigma)$ is given by
\[
||Q||_{L^2(\Sigma)}=
\left(\int\limits_\Sigma \rm{tr}(Q^\dag(z)
Q(z))|dz|\right)^{1/2}
=\left(\sum\limits_{j,l=1}^2\int\limits_{\Sigma}|Q_{jl}(z))|^2|dz|\right)^{1/2}.
\]
The operator ${\mathcal K}$  is defined by the jump matrix $J(t,x,z)$ and
the generalized function $$ \D \frac{1}{s-z_+}=
\lim_{z'\to
z,z'\in side+} \D \frac{1}{s-z'}.
$$
Furthermore, since the jump matrix $J(t,x,\la)$ has a positive definite
real part, when $\la\in\mathbb{R}$, then Theorem 9.3 from \cite{Zhou} (p.984)
guarantees the $L^2$ invertibility of the operator $Id-{\mathcal K}$
($Id$ is the identical operator). The function $R(t,x,z)$ belongs to
$L^2(\Sigma)$ because $I-J(t,x,z)\in L^2(\Sigma)$ when $z\in\Sigma$, and the Cauchy
operator
$$
C_+[f](z):=\D \frac{1}{2\pi\ii}\int\limits_\Sigma \D  {\frac{
f(s)}{s-z_+}ds}=\D \frac{f(z)}{2}+\text{p.v.}\D \frac{1}{2\pi\ii}\int\limits_\Sigma
\D {\frac{f(s)}{s-z}ds}
$$
is bounded in the space $L^2(\Sigma)$  \cite{LS}.
Therefore, the singular integral equation (\ref{K}) has a unique solution
$Q(t,x,z)\in L^2(\Sigma)$  for any  fixed $x, t\in\mathbb{R}$, and
formula (\ref{M}) gives the solution of the above RH problem.

\emph{Uniqueness.} A sketch of the proof is as follows.
Since $\det J(t,x,z)\equiv1$ one can find that $\det M(t,x,z)\equiv 1$ by repeating
step by step the proof of the Theorem 7.18 from \cite{D} (p.194-198). Hence the matrix $M^{-1}(t,x,z)$ exists, is analytic in $z\in\mathbb{C}\setminus\Sigma$ and bounded up to the contour $\Sigma$.
Let now suppose that there is another matrix $\tilde M(t,x,z)$, which solves the given Riemann-Hilbert problem. Then
$$
\tilde M_-(t,x,z)M_-^{-1}(t,x,z)=\tilde M_+(t,x,z)J(t,x,z)J^{-1}(t,x,z)
M_+^{-1}(t,x,z)$$$$=\tilde M_+(t,x,z)M_+^{-1}(t,x,z).
$$
Since $\tilde M(t,x,z)$ and $M^{-1}(t,x,z)$ are bounded up to the contour $\Sigma$, end points and points of self-intersection are removable singularities.
Hence the matrix $\tilde M(t,x,z)M^{-1}(t,x,z)$ is analytic in
$z\in\mathbb{C}$ and tends to identity matrix as $z\to\infty$. By Liovilles's theorem
$\tilde M(t,x,z)M^{-1}(t,x,z)\equiv I$ and therefore $\tilde M(t,x,z)\equiv M(t,x,z)$,
i.e. the matrix $M(t,x,z)$ is unique.
\end{proof}

Now we specialize the jump matrix $J(t,x,z)$ in such a way that corresponding RH problem solution $M(t,x,z)$  generates exactly the Maxwell-Bloch equations. Other specializations of the jump matrix $J(t,x,z)$ give another nonlinear integrable equations.
\begin{thm}\label{Phi}
Let inhomogeneous broadening $n(\la)$ be smooth and fast decreasing
for $\la\in\mathbb{R}$ and
$$
\eta(z)=z-\frac{1}{4}\int\limits_{-\infty}^\infty\frac{n(s)}{s-z}ds,
\qquad \int\limits_{-\infty}^\infty{n(s)}ds=\pm1.
$$
Let $\Phi(t,x,z):=M(t,x,z)e^{-\ii(zt-x\eta(z))\sigma_3}$, where $M(t,x,z)$ is the solution of the regular RH problem ($R1$)-($R4$) with the jump matrix $J(t,x,z)$ such that
\begin{itemize}
\item $J(t,x,z)=\begin{cases}e^{-\ii
(zt-x\eta(z))\sigma_3}\hat J(z)e^{\ii(zt-x\eta(z))\sigma_3}, \qquad z\in\Sigma\setminus\mathbb{R}\quad (\Im z\neq0),\\
e^{-\ii
(\la t-x\eta_+(\la))\sigma_3}\hat J(x,\la)e^{\ii(\la t-x\eta_-(\la))\sigma_3}, \qquad z=\la\in\mathbb{R};\end{cases}$
\item $\hat J(z)$ is independent on $t$, $x$ and satisfies the Schwartz reflection principle;
\item $\hat J(x,\la)$ is independent on $t$ and has a positive definite real part for any $x\in\mathbb{R}_+$ and $\la\in\mathbb{R}$;
\item $\hat J(x,\la)=I+\ord(\la^{-1}), \qquad \la\to\pm\infty.$
\end{itemize}
If $M(t,x,z)$ is absolutely continuous (smooth) in $t$ and $x$, then   $\Phi(t,x,z)$ ($z\in\mathbb{C}_\pm$) satisfies AKNS equations:
$$
\Phi_t = -(\ii z\sigma_3+H(t,x))\Phi, \qquad
\Phi_x =(\ii z\sigma_3+H(t,x)-\ii G(t,x,z))\Phi
$$
almost everywhere (point-wise) with respect to $t$ and $x$. These equations become
(\ref{teq}) and (\ref{pmxeq}) as $z=\la\pm\ii0$.  Matrix $H(t,x)$ is given by   
\begin{equation}\label{H}\nonumber
H(t,x)=-\ii[\sigma_3, m(t,x)],\qquad m(t,x)= \D\frac{1}{\pi}\int
\limits_{\Sigma}(I+Q(t,x,z))(J(t,x,z)-I)dz,
\end{equation}
where  $Q(t,x,z)$ is the unique solution of singular integral equation (\ref{K}).
Matrix $F(t,x,\lambda)$ is Hermitian and has the structure: $F(t,x,\lambda)=\begin{pmatrix}
  {\mathcal N}(t,x,\lambda) &\rho(t,x,\lambda) \\
  \rho^*(t,x,\lambda) &-{\mathcal N}(t,x,\lambda) \\
\end{pmatrix}$. It is reconstructed as the following jump:
$$
\D\frac{\pi n(\la)}{2}F(t,x,\la)=\Phi_x(t,x,z)\Phi^{-1}(t,x,z)\vert_{z=\la+\ii0}-
\Phi_x(t,x,z)\Phi^{-1}(t,x,z)\vert_{z=\la-\ii0}.
$$
\end{thm}

\begin{proof}
The matrix $\Phi(t,x,z)$ is analytic in $z\in\mathbb{C}\setminus\Sigma$ and has the
jump across $\Sigma$:
$$
\Phi_-(t,x,z)=\Phi_+(t,x,z)\hat J(x,z),
$$
where $\hat J(x,z)$ is independent on $t$. This relation implies:
$$
\frac{\partial\Phi_-(t,x,z)}{\partial t}\Phi^{-1}_-(t,x,z)=\frac{
\partial\Phi_+(t,x,z)}{\partial t}\Phi^{-1}_+(t,x,z)
$$
for $z\in\Sigma$. This  relation  means that the matrix
logarithmic derivative $\Phi_t(t,x,z)\Phi^{-1}(t,x,z)$ is
analytic (entire) in $z\in\mathbb{C}$. Indeed, matrices $M(t,x,z)$, $M^{-1}(t,x,z)$
and derivative (in $t$) $M_t(t,x,z)$  are analytic in
$z\in\mathbb{C}\setminus\Sigma$ and the Cauchy integral (\ref{M})
gives the following asymptotic formulas:
$$
M(t,x,z)=I+\frac{m_\pm(t,x)}{z} +\ord(z^{-2}),\qquad \frac{d
M(t,x,z)}{dt}=\frac{dm_\pm(t,x)/dt}{z}+O(z^{-2}), \quad z\in\mathbb{C}_\pm
$$
as $z\to\infty$. Hence
$$
\Phi_t(t,x,z)\Phi^{-1}(t,x,z)=-\ii z\sigma_3+\ii[\sigma_3,
m_+(t,x)]+\ord(z^{-1})
$$$$
 \hskip6cm=-\ii z\sigma_3+\ii[\sigma_3, m_-(t,x)]+\ord(z^{-1}),
\qquad z\to\infty,
$$
where $[A, B]:=AB-BA$ and
$$
m_-(t,x)=m_+(t,x)=m(t,x)=\frac{1}{\pi}\int\limits_\Sigma
(I+Q(t,x,z))(J(t,x,z)-I)d z
$$
Since $\Phi_t(t,x,z)\Phi^{-1}(t,x,z)$ has no jump across $\Sigma$ and $M_t(t,x,z)$, $M^{-1}(t,x,z)$ are bounded up to the boundary, then the end
points and points of self intersection of the contour $\Sigma$ are
removable singularities for $\Phi_t(t,x,z)\Phi^{-1}(t,x,z)$.
Therefore, by  Liouville's theorem, this derivative is a
polynomial:
$$
U(z):=\Phi_t(t,x,z)\Phi^{-1}(t,x,z)=-\ii z\sigma_3-H(t,x),
$$
where $H(t,x):=-\ii[\sigma_3,
m(t,x)]=\begin{pmatrix}0&q(t,x)\\p(t,x)&0\end{pmatrix}$. Using the
Schwartz symmetry reduction of the jump matrix $J(t,x,z)$, we
show (sf.\cite{FT}) that $U(z)=\sigma_2 U^*(z^*)\sigma_2$, where
$\sigma_2=\begin{pmatrix}0&-\ii\\\ii&0
\end{pmatrix}$. This reduction imply $H(t,x)= -H^\dag(x,t)$, i.e. $q(t,x)=-p^*(t,x)$ and we put $q(t,x):={\mathcal E}(t,x)/2$.
Formula (\ref{M}) gives an integral representation for the electric field envelope
\be\label{calE}{\mathcal E}(t,x)=\D \frac{2}{\pi}\int
\limits_\Sigma\left([I+Q(t,x,z)][J(t,x,z)-I]\right)_{12}dz.
\ee
through the solution $Q(t,x,z)$ of the singular integral equation (\ref{K}).
Thus $\Phi(t,x,z)$ satisfies equation (\ref{teq}), and a scalar function ${\mathcal E}(t,x)$ is defined by (\ref{calE}).

In contrast with previous case matrix logarithmic derivative
$\Phi_x(t,x,z)\Phi^{-1}(t,x,z)$ is analytic in $z\in\mathbb{C}_\pm$ only. Indeed, since matrix $\hat J(z)$ is independent on $t$ and $x$ for $z\in\Sigma\setminus\mathbb{R}$, then this logarithmic derivative is continuous across the contour $\Sigma\setminus\mathbb{R}$, while it is not continuous across the real line because $\hat J(x,\la)$ depends on $x$.
End points of the contour $\Sigma$ are removable singularities because  the matrices $M_x(t,x,z)$ and $M^{-1}(t,x,z)$ are bounded up the conjugation contour $\Sigma$.  Further, the asymptotic behavior  at infinity gives
$$
 \Phi_x(t,x,z)\Phi^{-1}(t,x,z)=\ii z\sigma_3 +H(t,x) +O(z^{-1}),
\qquad z\in\mathbb{C}_\pm, \quad z\to\infty.
$$
Therefore we find that $\D\Phi_x(t,x,z)\Phi^{-1}(t,x,z)-\ii z\sigma_3 -H(t,x)$ is represented as the Cauchy integral:
$$
\Phi_x(t,x,z)\Phi^{-1}(t,x,z)-\ii z\sigma_3 -H(t,x)=\D\frac{1}{4\ii}\int\limits_{-\infty}^\infty
\D\frac{F(t,x,s)n(s)}{s-z}ds, \qquad z\notin\mathbb{R},
$$
where $F(t,x,\la)$ is some matrix, and factor $1/4\ii$ is chosen for a convenience.
Due to reflection (symmetry) property of the jump matrix we find that $F(t,x,\la)$ is Hermitian. Since ${\rm tr} (\Phi_x(t,x,\la\pm\ii0)\Phi^{-1}(t,x,\la\pm\ii0)) =(\det\Phi(t,x,\la\pm\ii0))^\prime_x\equiv0$ and ${\rm tr}\sigma_3={\rm tr} H(t,x)=0$ then ${\rm tr}F(t,x,\la)=0$ and, hence, $F(t,x,\la)$ has the structure:
$$
F(t,x,\la):=\begin{pmatrix}{\mathcal N}(t,x,\la)&{\mathcal\rho}(t,x,\la)\\
{\mathcal\rho}^*(t,x,\la)&-{\mathcal N}(t,x,\la)\end{pmatrix}.
$$

Thus we see that the matrix $\Phi(t,x,z)$ satisfies two differential equations:
\begin{align}\label{teq1}
\Phi_t &=U(t,x,z)\Phi, \qquad U(t,x,z)=-\ii z\sigma_3-H(t,x) \\
\Phi_x &=V(t,x,z)\Phi, \qquad V(t,x,z)=\ii z\sigma_3+H(t,x)-\ii G(t,x,z),\label{xeq1}
\end{align}
where
$$
G(t,x,z)=\D\frac{1}{4}\int\limits_{-\infty}^\infty  \D\frac{F(t,x,s)n(s)}{s-z}ds, \qquad z\notin\mathbb{R}.
$$
For real $z=\la\in\mathbb{R}$ we have two differential in $x$ equations:
\be\label{pmxeq1}
\Phi_x =V_\pm(t,x,\la)\Phi, \qquad V_\pm(t,x,\la)=\ii\la\sigma_3+H(t,x)-
\ii G_\pm(t,x,\la),
\ee
where $G_\pm(t,x,\la):=G(t,x,\la\pm\ii0)$. In particular, the last equations give
a reconstruction of the function $F(t,x,\la)$:
$$
\D\frac{\pi n(\la)}{2}F(t,x,\la)=\Phi_x(t,x,\la+\ii0)\Phi^{-1}(t,x,\la+\ii0)-
\Phi_x(t,x,\la-\ii0)\Phi^{-1}(t,x,\la-\ii0),
$$
where $\Phi(t,x,z)=M(t,x,z)e^{-\ii(zt-x\eta(z))\sigma_3}$, and $M(t,x,z)$ is the solution of the regular RH problem ($R1$)-($R4$). The compatibility condition $(\Phi_{xt}(t,x,\la\pm\ii0)=\Phi_{tx}(t,x,\la\pm\ii0))$
gives the identity in $\la$:
\begin{eqnarray*}
U_x(t,x,\la)-V^\pm_t(t,x,\la)+[U(t,x,\la),V^\pm(t,x,\la)]=0.
\end{eqnarray*}
This identity is equivalent to
\begin{eqnarray*}
H_t(t,x)+H_x(t,x)-\D\frac{1}{4}\int\limits_{-\infty}^\infty  [\sigma_3,F(t,x,s)]n(s)ds=\\=
\D\frac{\ii}{4}\int\limits_{-\infty}^\infty\D\frac{(F_t(t,x,s)+[\ii s\sigma_3+H(t,x),
F(t,x,s)]}
{s-\la\mp\ii0}n(s)ds
\end{eqnarray*}
and it is possible if and only if the left and right hand sides are equal zero, i.e.
\begin{eqnarray*}
H_t(t,x)+H_x(t,x)-\D\frac{1}{4}\int\limits_{-\infty}^\infty  [\sigma_3,F(t,x,s)]n(s)ds&=0\\
F_t(t,x,\la)+[\ii\la\sigma_3+H(t,x),F(t,x,\la)]&=0.
\end{eqnarray*}
These matrix equations are equivalent to the MB equations (\ref{MB1})-(\ref{MB3}).
Thus we proved that the matrices $\Phi(t,x,\la\pm\ii0)$ satisfy equations (\ref{teq1}) and (\ref{pmxeq1}),
which coincide with AKNS system (\ref{teq}) and (\ref{pmxeq}), and
scalar functions ${\mathcal E}(t,x)$, ${\mathcal N}(t,x,\la)$ and ${\mathcal\rho}(t,x,\la)$
satisfy the Maxwell-Bloch equations (\ref{MB1})-(\ref{MB3}).
\end{proof}

\section{Conclusions}

It is proved that the mixed problem (\ref{ic})-(\ref{bc}) to the Maxwell-Bloch equations (\ref{MB1})-(\ref{MB3}) is linearizable completely by using the matrix Riemann-Hilbert problem $RH1$ - $RH5$ or $RRH1$ - $RRH3$. More general Riemann-Hilbert problem $R1$ - $R3$ generates a solution to the Maxwell-Bloch equations if conjugation contour and jump matrix are satisfied the Schwartz reflection principal and some special restrictions. Among generated solutions there are: solutions defined for $t\in\mathbb{R}$ and $x\in\mathbb{R}_+$ and studied in \cite{AKN}, \cite{GZM}, solutions to the mixed problem (\ref{MB1})-(\ref{MB3}), (\ref{ic})-(\ref{bc}) ($t,x\in\mathbb{R}_+$) with decreasing or periodic input pulse ${\mathcal E}(t,0)$ and different initial functions ${\mathcal E}(0,x)$, ${\mathcal N}(0,x)$, $\rho(0,x)$, etc. The kind of solutions is defined by the specialization of the conjugation contour and the jump matrix. Suggested matrix RH problems will be useful for studying the long time/long distance asymptotic behavior of solutions to the MB equations using the Deift-Zhou method of steepest decent.



\end{document}